\documentclass[11pt]{article}
\usepackage{fullpage}
\usepackage{amsmath}
\usepackage{amsfonts}
\usepackage{amsthm}
\usepackage{amssymb}
\usepackage{dsfont}
\usepackage{graphicx}
\usepackage{caption}
\usepackage{subcaption}
\usepackage{hyperref}
\usepackage{mathabx}
\usepackage{mathrsfs}
\usepackage{color}
\newtheorem{theorem}{Theorem}
\newtheorem{lemma}[theorem]{Lemma}
\newtheorem{claim}[theorem]{Claim}
\newtheorem{cor}[theorem]{Corollary}
\newtheorem{definition}{Definition}
\newtheorem{example}{Example}
\newtheorem{remark}{Remark}

\newcommand{\di}{\mathrm{d}}

\DeclareMathOperator*{\argmax}{\arg\max}

\allowdisplaybreaks
\title{The Welfare Gap of Strategic Storage:\\ Universal Bounds and Price Non-Linearity}

\author{Zhile Jiang \thanks{Department of Computer Science, Aarhus University, Denmark. Email: {\tt zhile@cs.au.dk}.} \and Xinhao Nie \thanks{Department of Computer Science, Aarhus University, Denmark. Email: {\tt nie@cs.au.dk}.} \and Stratis Skoulakis \thanks{Department of Computer Science, Aarhus University, Denmark. Email: {\tt stratis@cs.au.dk}.}}

\date{}

\begin{document}

\maketitle

\begin{abstract}
This paper studies the efficiency of battery storage operations in electricity markets by comparing the social welfare gain achieved by a central planner to that of a decentralized profit-maximizing operator. The problem is formulated in a generalized continuous-time stochastic setting, where the battery follows an adaptive, non-anticipating policy subject to periodicity and general convex constraints. We quantify the efficiency loss by bounding the ratio of the optimal welfare gain to the gain under profit maximization. First, for linear price functions, we prove that this ratio is tightly bounded by $4/3$. We show that this bound is a structural invariant: it is robust to arbitrary stochastic demand processes and accommodates general convex operational constraints. Second, we demonstrate that the efficiency loss can be unbounded for general convex price functions even in a canonical discrete-demand benchmark, so convexity alone is insufficient to guarantee market efficiency. Third, within the same benchmark we analyze monomial price functions, where the degree controls the curvature, and prove that the loss grows with the degree yet remains bounded by $2$. Finally, we extend the linear analysis to $n$ competing batteries, where a potential-game argument gives a unique equilibrium and an efficiency loss that decreases to $1$ as the number of batteries grows.

\end{abstract}

\newpage

\section{Introduction}

Battery Energy Storage Systems (BESS) are increasingly central to electricity markets, smoothing the volatility introduced by intermittent renewable generation~\cite{GRS16, SY21, UHLE13} and high-load demand sectors~\cite{CWHetal25, BBM23}. BESS operators exploit temporal arbitrage, absorbing surplus energy during off-peak periods and injecting it during peak demand, thereby reducing price fluctuations~\cite{FFTFGSS16, ZAPTH23}. However, a fundamental tension arises: these operators are profit-maximizing agents that capitalize on price differentials, whereas the grid operator aims to minimize total generation cost. This misalignment between private and social objectives naturally invites a Price of Anarchy (PoA) analysis. While extensive literature addresses operational optimization of individual BESS units~\cite{HARBB22, SGMD22, THMW20, WM21, ZAPTH23}, formal PoA analysis of strategic storage has only recently begun.

Anunrojwong et al.~\cite{ABBX25} took the first step, formulating the battery storage problem as a game between strategic price-making batteries and the market. In a discrete two-period model with stochastic demand and linear supply curves, they derived closed-form equilibrium strategies for both the monopoly and the $n$-battery Cournot settings. For a single battery they proved the PoA lies between $9/8$ and $4/3$ depending on market parameters, and for $n$ competing batteries they showed that a unique Cournot equilibrium exists and that the PoA converges to~$1$ at rate $1/n^2$. Their analysis, however, is confined to linear pricing in a two-period structure and does not incorporate general operational constraints on the battery.

These modeling restrictions leave open several natural questions: whether the $4/3$ worst-case PoA persists under continuous-time dynamics and binding operational constraints, how the efficiency loss behaves across non-linear price functions, and whether the equilibrium and convergence structure of the multi-battery game extends to richer, infinite-dimensional strategy spaces.

We resolve these questions within a continuous-time stochastic framework that accommodates arbitrary demand distributions and general convex operational constraints. Our results establish that the $4/3$ PoA is a structural invariant of linear pricing, delineate the sharp boundary between pricing regimes with bounded and unbounded efficiency loss, and characterize how competition restores efficiency in an infinite-dimensional strategy space.

\subsection{Our Contribution}

We summarize our main results here. Formal statements appear in Sections~\ref{sec:linear}--\ref{sec:multi-battery}, within the general framework of Section~\ref{sec:model}.

For linear price functions, we establish a tight $PoA = 4/3$ (Theorem~\ref{thm:PoA-linear}) for arbitrary stochastic demand and any convex constraint set $\Omega$. The matching upper bound of $4/3$ was previously known only for the two-period unconstrained model of~\cite{ABBX25}. Our result shows it is a structural invariant of linear pricing that persists in continuous time and under arbitrary convex operational constraints, independent of the demand process. The proof uses a variational inequality in a function-space inner product that isolates the gap between the welfare and revenue objectives.

For general convex price functions, we prove that the PoA can be unbounded even in the simplest deterministic setting $\mathcal{I}_{\text{step}}$ (Theorem~\ref{thm:PoA-counterexample}), and we extend this to convex polynomials through a Bernstein approximation argument (Corollary~\ref{cor:poly-lb}). Linearity is therefore a sharp boundary, since any departure into general convexity can destroy bounded efficiency guarantees.

For monomial price functions $P(z)=\alpha z^d$ within $\mathcal{I}_{\text{step}}$, we prove a universal upper bound $PoA \leq 2$ for all degrees~$d$ (Theorem~\ref{thm:PoA-mono-ub}), a tight bound of $27/19$ for the quadratic case $d=2$ (Theorem~\ref{thm:PoA-mono-2-ub}), and degree-dependent lower bounds approaching $e/(e-1)$ as $d\to\infty$ (Theorem~\ref{thm:PoA-mono-lb}). The PoA thus increases with the degree of non-linearity but stays bounded, in contrast to the unboundedness for general convex functions.

Extending to $n$ identical batteries competing under linear pricing, we prove that the game admits a strictly concave potential function, yielding a unique pure Nash equilibrium that is necessarily symmetric, with $PoA=(n+1)^2/(n(n+2))$ (Theorem~\ref{thm:PoA-linear-n}). This decreases monotonically to $1$ as $n\to\infty$, so competition fully restores efficiency in the limit. The two-period analysis of~\cite{ABBX25} established a unique Cournot equilibrium and $PoA\to 1$. Our potential-game argument applies in the continuous-time function space and yields an exact closed-form PoA.

\subsection{Other Related Works}

\noindent \textbf{Strategic behavior.} Our model treats the battery as a single strategic price-maker, reflecting the market power that storage operators possess due to high market concentration and transmission constraints~\cite{BBW02}. The theoretical literature on strategic storage has established that profit-maximizing operators provide less price smoothing than a social planner and can exploit inventory constraints to enhance market power~\cite{CM01, Bushnell03}. Sioshansi~\cite{Sioshansi10} shows that such operators have incentives to withhold capacity or alter dispatch to manipulate price spreads, and that without regulation, private storage can reduce social welfare by amplifying price volatility~\cite{Sioshansi14}. Empirical studies further quantify this efficiency gap~\cite{BBCKLW23, MohsenianRad15, ZLHN24}.

\noindent \textbf{Pricing mechanisms.} Wholesale electricity markets mostly operate under a Pay-as-Clear mechanism, where the market price equals the marginal cost of the marginal generator~\cite{FR03}. The total social generation cost is then the integral of the price function. Linear price-response models are commonly adopted for tractability~\cite{ABBX25, Bushnell03}, but empirical evidence shows significant convexity in the price--load relationship, with marginal costs rising disproportionately near supply limits~\cite{KR05, LW04, BL02}.

\noindent \textbf{Demand modeling.} The literature on stochastic demand modeling decomposes load into deterministic seasonal trends and stochastic fluctuations~\cite{Taylor03, FH11}. Early works used diffusion models for pricing dynamics~\cite{Barlow02, BBK08}, and we adopt their continuous framework for the demand process. Our abstraction over general probability distributions ensures robustness to modern forecasting methods, from density forecasting~\cite{HF09, HF16} to deep autoregressive networks~\cite{SFGJ20}.

\noindent \textbf{Operational constraints.} Battery operational constraints, including power ratings, energy capacity, and ramp rates, are well modeled as convex constraints or linear relaxations~\cite{MohsenianRad15, THMW20}, with established numerical frameworks for optimal scheduling under uncertainty~\cite{LM10, AB17}. Continuous-time extensions are developed in~\cite{GTL12, CFGZ19}. Since the control problem is well understood, we abstract away operational details to focus on economic efficiency.

\section{Model and Notations}
\label{sec:model}

\subsection{Market Model}

We model the electricity market over a normalized time period $T=[0,1]$, which represents one cycle of a periodic market. The demand $D=\{D(t)\}_{t\in[0,1]}$ is a stochastic process on $(\Xi,\mathcal{F},\mathbb{P})$ with filtration $\mathcal{F}=\{\mathcal{F}_t\}_{t\in[0,1]}$, normalized so that $D(t)\in[0,1]$ almost surely. We assume the demand is periodic with period $1$, so that the law of $D$ repeats across cycles. This is standard for daily electricity operations and lets us analyze a single representative cycle, and it motivates the cyclic battery constraint introduced below. The setting encompasses deterministic demand as a special case. See Example~\ref{example:demand} and Fig.~\ref{fig:demand_process}.

\begin{example}
\label{example:demand}
Consider the demand $D$ defined by the following stochastic process:
\begin{align*}
D(t)=m(t)+\sigma(t)\cdot Z(t),
\end{align*}
where the deterministic component $m(t)=0.7+0.2\sin(2\pi t)$ characterizes the baseline diurnal trend of the load. To capture the temporal correlation inherent in the power system, arising from physical inertia and gradual load changes, the noise term $Z(t)$ is modeled as a non-Markovian process derived from a Gaussian Process with a squared-exponential kernel. This specification ensures that the demand trajectories exhibit realistic smoothness (differentiability) rather than erratic, discontinuous jumps. The term $\sigma(t)$ represents time-varying volatility, strictly constrained to ensure the realized demand remains within physical bounds.
\end{example}

\begin{figure}[h]
    \centering
    \includegraphics[width=0.6\linewidth]{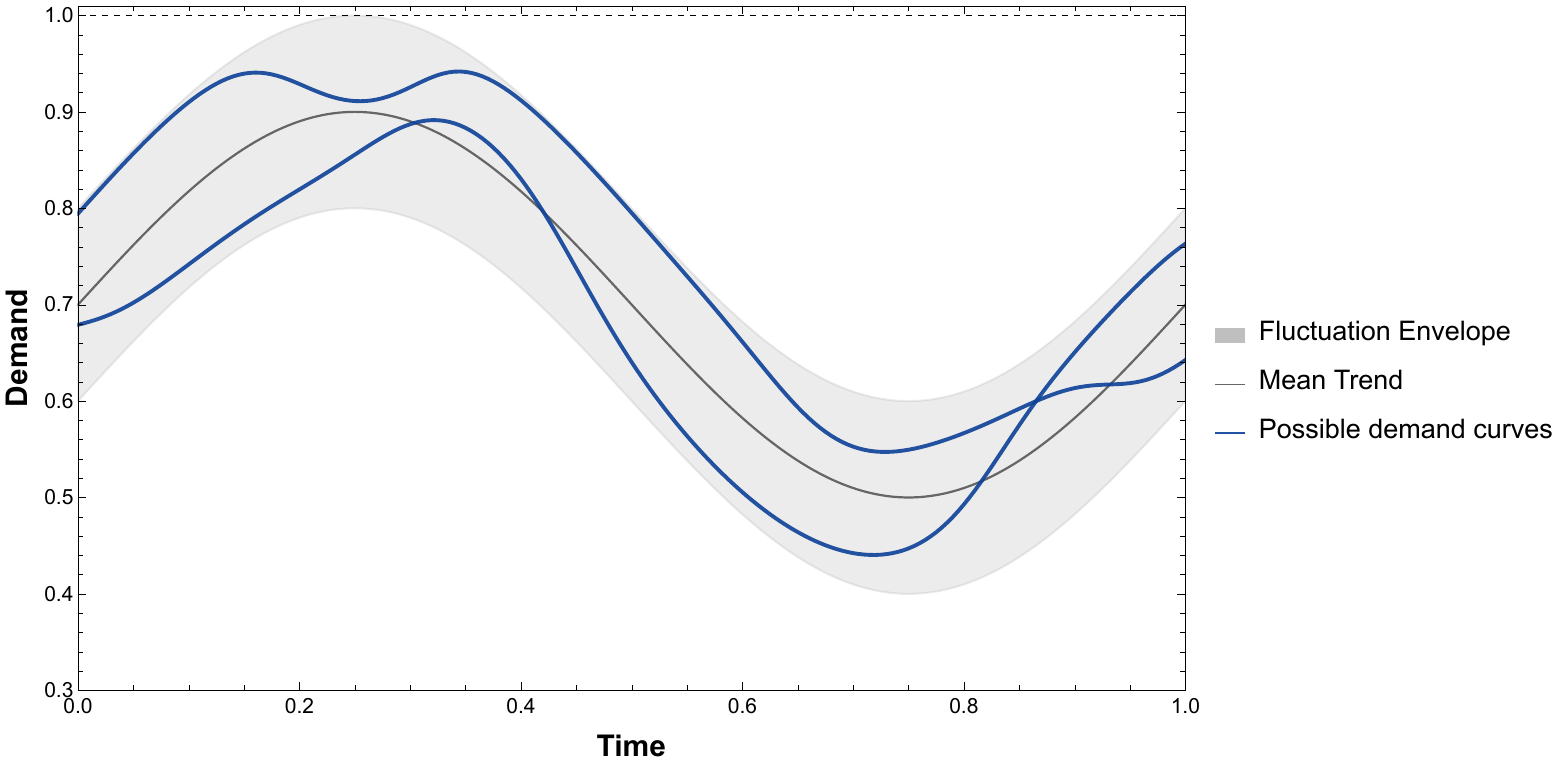}
    \caption{The figure shows the stochastic process defined in Example \ref{example:demand}. The gray dashed line depicts the deterministic mean trend. The solid blue curves show two realizations of the stochastic process, while the gray shaded region defines the fluctuation envelope.}
    \label{fig:demand_process}
\end{figure}

A battery operates via a rate schedule $B=\{B(t)\}_{t\in[0,1]}$, with $B>0$ for discharging and $B<0$ for charging. The policy must be \textit{non-anticipating}: $B(t)$ depends only on the demand history up to time $t$. Writing $B(t,\xi)$ and $D(t,\xi)$ for scenario $\xi\in\Xi$:
\begin{align}
\label{constraint:Non-anticipating}
B(t, \xi) = B(t, \xi'), \quad \forall \xi, \xi' \in \Xi \text{ such that } D(s, \xi) = D(s, \xi') \text{ for almost all } s \in [0, t].
\end{align}
The set of non-anticipating policies is convex. We suppress $\xi$ hereafter. The net demand $N(t)=D(t)-B(t)$ must remain in $[0,1]$:
\begin{align}
\label{constraint:Box}
\Pr_{\xi\sim\Xi}\left[0 \leq D(t,\xi)-B(t,\xi) \leq 1,\forall t\in[0,1]\right]=1.
\end{align}
The battery must also return to its initial charge level over the cycle:
\begin{align}
\label{constraint:Periodicity}
\Pr_{\xi\sim\Xi}\left[\int_0^1 B(t,\xi) \mathrm{d} t = 0\right]=1.
\end{align}
We denote by $\mathcal{B}$ the set of all policies satisfying \eqref{constraint:Non-anticipating}, \eqref{constraint:Box}, and \eqref{constraint:Periodicity}.

Additional technological limitations are modeled via a feasible set $\Omega$, with $\mathcal{B}\cap\Omega$ convex and containing the zero function. Examples include:

\begin{itemize}
    \item The power constraint: The power (charging/discharge rate) is bounded by $\gamma$, i.e., 
    \begin{align*}
        |B(t)| \leq \gamma\quad \forall t \in [0,1].
    \end{align*}

    \item The energy capacity constraint: The total energy stored or released between any two time points cannot exceed the physical capacity $c > 0$, i.e., 
    \begin{align*}
        \left|\int_{t_1}^{t_2} B(t) dt\right| \leq c\quad \forall t_1, t_2 \in [0,1].
    \end{align*}
    Notice that the above inequalities are equivalent to the existence of an initial charge that ensures the state-of-charge of the battery is between $[0,c]$.   
    \item The ramp rate constraint: The rate of change of the battery power is limited by $\delta$, i.e., 
    \begin{align*}
        |B'(t)| \leq \delta\quad \forall t \in [0,1].
    \end{align*}
\end{itemize}

All listed constraints are convex, and our framework accommodates any convex $\Omega$.

\subsection{Pricing, Welfare, and Efficiency Metrics}

With the physical model in place, we now introduce the economic side: how electricity is priced, how we measure the social benefit of battery operation, and how we quantify the efficiency loss when a battery maximizes profit rather than social welfare.

The price of electricity at time $t$ is determined by an increasing price function $P: [0,1] \to \mathbb{R}_{\geq 0}$, mapping the net demand $N(t)$ to a market-clearing price. The social generation cost is $G(N) = \int_{0}^{N} P(x) \,\mathrm{d} x$. This formulation assumes a Pay-as-Clear (or Uniform Pricing) mechanism, widely used in practice, under which all dispatched generators are compensated at the price set by the marginal unit. In a theoretically idealized market with a continuum of infinitesimal generators, the unique equilibrium has all generators bidding their true marginal costs, so $P(N)$ represents the system's marginal cost curve and $G(N)$ the total generation cost at demand level $N$.

Let $I = (D, P, \Omega)$ denote a specific market instance. Since demand is stochastic, we measure performance in expectation. The key welfare metric is the \textit{expected improvement on social cost} achieved by a battery policy $B$, which captures how much the battery reduces the total generation cost relative to the no-battery baseline:
\begin{align*}
    \mathrm{WEL}_I(B)\overset{\text{def}}{=}\mathbb{E}\left[ \int_{0}^{1} G(D(t))-G(D(t) - B(t)) \mathrm{d} t \right].
\end{align*}

The central tension in our model is between two operating regimes. A social planner would operate the battery to minimize generation cost, while a private operator seeks to maximize arbitrage profit. These objectives are generally misaligned, and we formalize both below.
\begin{description}
    \item[Centralized Battery] A system operator controls the battery to maximize the expected improvement in social cost:
    \begin{align}\label{obj:CB}
    \max_{B \in \mathcal{B}\cap\Omega} \mathrm{WEL}_I(B).
    \end{align}

    \item[Decentralized Battery] A profit-maximizing entity operates the battery to maximize its expected arbitrage revenue:
    \begin{align}\label{obj:DCB}
    \max_{B \in \mathcal{B}\cap\Omega} \mathrm{REV}_I(B)=\mathbb{E}\left[ \int_{0}^{1} B(t) \cdot P(D(t) - B(t)) \mathrm{d} t \right].
    \end{align}
    Here $B(t)\cdot P(D(t)-B(t))$ is the instantaneous profit from trading power at rate $B(t)$ at the endogenous price $P(D(t)-B(t))$, which the battery influences through its own action.
\end{description}

We assume throughout that the maxima in \eqref{obj:CB} and \eqref{obj:DCB} are attained, which holds in all instances we study, where the objectives are concave over the convex feasible set $\mathcal{B}\cap\Omega$.

\begin{remark}
    In the special case where the demand $D$ is deterministic, the expectation operator $\mathbb{E}[\cdot]$ can be removed. The adaptive nature of $B$ simplifies to a standard trajectory optimization problem where the entire demand curve $D_{[0,1]}$ is known a priori. In this setting, the problem reduces to the classical deterministic formulation where constraints and objectives are evaluated on a single realized path.
\end{remark}

Let $\mathcal{B}_{\text{CB}}(I)$ and $\mathcal{B}_{\text{DCB}}(I)$ be the sets of optimal solutions to \eqref{obj:CB} and \eqref{obj:DCB}, respectively. The \textit{Price of Anarchy} (PoA) measures how much welfare is lost due to selfish operation. It is defined as the ratio of the welfare gain under centralized control to that under profit maximization:

\begin{align*}
PoA(I) \overset{\text{def}}{=} 
\begin{cases}
    \max\frac{\mathrm{WEL}_I(B_{\text{CB}})}{\mathrm{WEL}_I(B_{\text{DCB}})}, &\text{if }\min \mathrm{WEL}_I(B_{\text{DCB}})>0 \\
    1, &\text{if }\min \mathrm{WEL}_I(B_{\text{DCB}})=0 \text{ and } \max \mathrm{WEL}_I(B_{\text{CB}})=0\\
    \infty, &\text{if }\min \mathrm{WEL}_I(B_{\text{DCB}})\leq 0
\end{cases}.
\end{align*}

The definition handles three cases for mathematical rigor. The first and most important case is when both batteries improve social cost, where the PoA captures the welfare gap. Since the zero function is always feasible, the centralized welfare is non-negative. We take the maximum over decentralized solutions to capture the worst-case welfare loss when there are ties. The second case ($PoA=1$) covers trivial instances where no battery can improve welfare. The third case ($PoA=\infty$) arises when the profit-maximizing battery actually harms social welfare.

For a class of instances $\mathcal{I}$, the Price of Anarchy is defined as $PoA(\mathcal{I}) \overset{\text{def}}{=} \sup_{I \in \mathcal{I}} PoA(I)$. We omit the dependence on $I$ or $\mathcal{I}$ unless it is needed for clarity.

Before proceeding, we provide Example \ref{exp:battery} and Fig. \ref{fig:battery} to clarify the concepts.
\begin{figure}[h]
    \centering
    \includegraphics[width=0.98\linewidth]{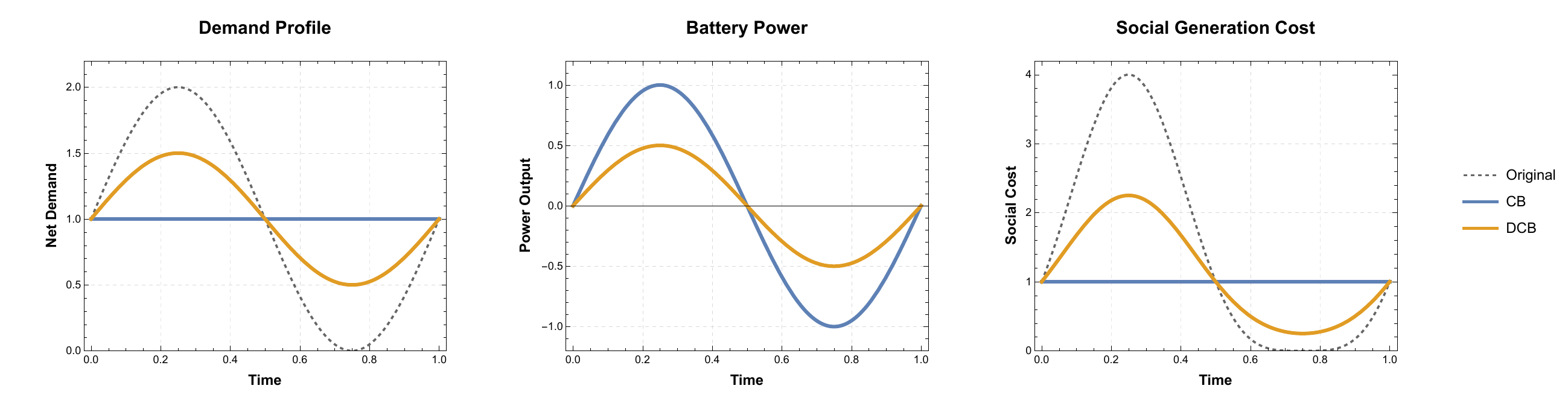}
    \caption{Illustration of Example \ref{exp:battery}. The left panel compares the original demand with the net demand adjusted by Centralized (CB) and Decentralized (DCB) batteries. The middle panel depicts the corresponding battery operation policies, while the right panel shows the instantaneous social generation costs before and after adjustment. Note that the revenue of the social welfare maximizing solution is zero, since the demand is completely smoothed.}
    \label{fig:battery}
\end{figure}
\begin{example}
    \label{exp:battery}
    Consider a deterministic demand $D(t)=1+\sin(2\pi t)$ and a linear price function $P(z)=2z$ giving a social generation cost function $G(z)=z^2$. Assuming there are no structural constraints, the centralized battery can minimize the social generation cost by setting $B_{CB}(t)=\sin(2\pi t)$, which improves the social generation cost by 
    \begin{align*}
        \mathrm{WEL}(B_{CB})=\int_0^1 (1+\sin(2\pi t))^2 \di t-\int_0^1 1^2\di t=1/2.
    \end{align*} 
    And the decentralized battery will set $B_{DCB}(t)=\sin(2\pi t)/2$ to maximize the revenue. The improvement of social cost is \begin{align*}
        \mathrm{WEL}(B_{DCB})=\int_0^1 (1+\sin(2\pi t))^2 \di t-\int_0^1 \left(1+\frac{\sin(2\pi t)}{2}\right)^2\di t=3/8.
    \end{align*}
    In this case, $PoA=\mathrm{WEL}(B_{CB})/\mathrm{WEL}(B_{DCB})=4/3$.
\end{example}

\section{Linear Price Functions}
\label{sec:linear}

We begin with linear price functions, a standard assumption in the literature. For this class, we establish a tight, universal bound on the Price of Anarchy that holds for any demand distribution and any convex constraint set.

\begin{theorem}
\label{thm:PoA-linear}
    Define the family of instances $\mathcal{I}_{\text{lin}}$ as the set of all market configurations $I = (D, P, \Omega)$ such that the price function $P(z)=a\cdot z+b$ where $a>0,b\geq 0$. We have $PoA(\mathcal{I}_{\text{lin}})= 4/3$.
\end{theorem}

This result substantially generalizes the findings in \cite{ABBX25}, which derived the $4/3$ bound only for piecewise constant demand with unconstrained operations. Our proof reveals that this bound is a structural invariant of linear markets: it stems solely from the misalignment between the social cost and revenue objectives, and is unaffected by the stochasticity of demand or the geometry of the feasible set. We first define some notation.

\begin{definition}
    For any stochastic process $X$ and $Y$ on probability space $(\Xi,\mathcal{F},\mathbb{P})$, we define 
    \begin{align*}
        \langle X, Y\rangle \overset{\text{def}}{=} \int_{\xi\in\Xi}\left[\int_0^1 X(t,\xi)\cdot Y(t,\xi)\di t\right]\di \mathbb P(\xi)=\mathbb{E}\left[\int_0^1 X(t)\cdot Y(t)\di t\right].
    \end{align*}
    And we denote $\|X\|^2\overset{\text{def}}{=}\langle X,X\rangle$. 
\end{definition}

The operation $\langle\cdot, \cdot\rangle$ is a generalized inner product since taking expectation and integration are both linear operations.

\noindent \emph{Proof of Theorem \ref{thm:PoA-linear}.}
We consider an instance $I=(D,P,\Omega)\in\mathcal{I}_{lin}$ with $P(x)=a\cdot x+b$ where $a>0$. By direct computation,
\begin{align}
    \mathrm{WEL}(B)&=\mathbb{E}\left[ \int_{0}^{1} \frac{1}{2}a\cdot(D(t))^2+b\cdot(D(t))-(\frac{1}{2}a\cdot(D(t)-B(t))^2+b\cdot(D(t)-B(t))) \di t \right]\notag\\
    &=\frac{1}{2}a\cdot(2\cdot\langle D,B \rangle-\|B\|^2),\label{eq:W-linear}
\end{align}
where the last equality uses Constraint \eqref{constraint:Periodicity}. Similarly,
\begin{align}
    \mathrm{REV}(B)=\mathbb{E}\left[ \int_{0}^{1} B(t) \cdot P(D(t) - B(t)) \mathrm{d} t \right]=a\cdot(\langle D,B\rangle-\|B\|^2)+\underbrace{b\cdot\mathbb{E}\left[\int_0^1 D(t)\di t\right]}_{\text{Constant term}}.\label{eq:REV-linear}
\end{align}

Define $B_\text{CB}$ and $B_\text{DCB}$ as optimal solutions to Programming \eqref{obj:CB} and \eqref{obj:DCB}. We have
\begin{align*}
    B_{\text{CB}}&\in\argmax_{B\in\mathcal{B}\cap \Omega} 2\cdot\langle D,B \rangle-\|B\|^2,\\
    B_{\text{DCB}}&\in\argmax_{B\in\mathcal{B}\cap \Omega} \langle D,B \rangle-\|B\|^2.
\end{align*}

Since the objectives are concave and differentiable and the feasible set $\mathcal{B}\cap \Omega$ is convex, the first-order optimality conditions are given by variational inequalities.

\begin{lemma}[Optimality Conditions]
\label{lemma:variational_inequalities}
Let $f$ be a concave and Fr\'echet differentiable functional on a convex set $\mathcal{X}$. If $x^* \in \argmax_{x\in\mathcal{X}} f(x)$, then $\langle \nabla f(x^*), x - x^* \rangle \leq 0$ for all $x \in \mathcal{X}$.
\end{lemma}

Applying Lemma \ref{lemma:variational_inequalities} to $B_{\text{DCB}}$ with objective $F_1(B) = \langle D, B \rangle - \|B\|^2$ and gradient $\nabla F_1(B) = D - 2\cdot B$ gives, for any $B \in \mathcal{B}\cap\Omega$,
\begin{align}
\label{inq:opt-1-app}
\langle D - 2\cdot B_{\text{DCB}}, B - B_{\text{DCB}} \rangle \leq 0.
\end{align}

This implies, for any optimal solution $B_{CB}$ and $B_{DCB}$,
\begin{align*}
    &4\cdot(2\langle D, B_{\text{DCB}} \rangle - \|B_{\text{DCB}}\|^2) - 3\cdot(2\langle D, B_{\text{CB}} \rangle - \|B_{\text{CB}}\|^2) \\
    &\geq 2\cdot\langle D, B_{\text{DCB}}\rangle+6\cdot(2\cdot(\|B_{\text{DCB}}\|^2-\langle B_{\text{DCB}}, B_{\text{CB}}\rangle)) - 4\cdot\|B_{\text{DCB}}\|^2 +3\cdot\|B_{\text{CB}}\|^2\\
    &\geq 12\cdot\|B_{\text{DCB}}\|^2-12\cdot\langle B_{\text{DCB}}, B_{\text{CB}}\rangle+3\cdot\|B_{\text{CB}}\|^2
    =3\cdot\|2\cdot B_{\text{DCB}}-B_{\text{CB}}\|^2\geq 0.
\end{align*}

The first inequality holds since substituting $B = B_{\text{CB}}$ into \eqref{inq:opt-1-app} gives
$\langle D, B_{\text{DCB}}\rangle-\langle D, B_{\text{CB}}\rangle \geq 2(\|B_{\text{DCB}}\|^2-\langle B_{\text{DCB}}, B_{\text{CB}}\rangle)$.
The second inequality holds since substituting $B = 0$ gives
$\langle D, B_{\text{DCB}}\rangle\geq 2\|B_{\text{DCB}}\|^2$.

Recall $B=0\in\mathcal{B}\cap \Omega$, so $\langle D,B_{\text{CB}} \rangle-\|B_{\text{CB}}\|^2\geq0$. If both objectives are zero, then $PoA=1$. Otherwise $\langle D,B_{\text{DCB}} \rangle-\|B_{\text{DCB}}\|^2>0$, and the inequality gives
$$PoA=\max\frac{2\langle D,B_{\text{CB}} \rangle-\|B_{\text{CB}}\|^2}{2\langle D,B_{\text{DCB}} \rangle-\|B_{\text{DCB}}\|^2}\leq \frac{4}{3}.$$
\qed

\section{Non-Linear Price Functions}
\label{sec:nonlinear}

The preceding section established a tight PoA of $4/3$ for linear pricing. We now investigate what happens beyond the linear regime. Following \cite{ABBX25}, the results below are derived within a canonical benchmark: deterministic two-level demand with a discrete battery response and no structural constraints.

\begin{definition}
\label{def:step}
        Define the family of instances $\mathcal{I}_{\text{step}}$ as the set of all market configurations $I$ satisfying the following properties:
    \begin{itemize}
        \item[-] The deterministic demand $D$ is a 2-piecewise constant (step) function, i.e.,
        \begin{align*}
        D(t)=\begin{cases}
            D_1, &t\in[0,t_1]\\
            D_2, &t\in(t_1,1]
        \end{cases}.
        \end{align*}
        \item[-] The battery operation policy $B$ is also a 2-piecewise constant function following the demand, i.e.,
        \begin{align*}
        B(t)=\begin{cases}
            B_1, &t\in[0,t_1]\\
            B_2, &t\in(t_1,1]
        \end{cases}.
        \end{align*}
    \end{itemize}
\end{definition}

Restricting the battery to a 2-piecewise response is without loss for step demand: both the revenue and welfare integrands are pointwise functions of $B(t)$ on each constant-demand interval, coupled only through the periodicity constraint~\eqref{constraint:Periodicity}.

\subsection{General Convex Price Functions}

We first show that linearity is a necessary condition for bounded efficiency: even within $\mathcal{I}_{\text{step}}$, general convexity can lead to unbounded efficiency losses.

\begin{theorem}
    \label{thm:PoA-counterexample}
    There exists at least one instance $I\in\mathcal{I}_{\text{step}}$,
    with a convex price function such that $PoA(I)$ is unbounded.
\end{theorem}

We begin with a structural claim about centralized optima that will be used throughout the remaining analysis. The proof is in Appendix \ref{appendix:proof-clm:cb-average}.

\begin{claim}
    \label{clm:cb-average}
     For an instance $I$ with deterministic demand and no structural constraints, the improvement of social cost is maximized when the battery balances the net demand to the average demand over $t\in[0,1]$, i.e., the maximum improvement is $\mathrm{WEL}(B_{CB})=\int_0^1 G(D)-G(\widebar{D})$, where $\widebar{D}=\int_0^1 D(t)\di t$ is the average demand and $B_{CB}(t)=D(t)-\widebar{D}$ is the optimal solution achieves this improvement.
\end{claim}
\noindent \emph{Proof sketch of Theorem \ref{thm:PoA-counterexample}.} Consider the demand $D(t)=1$ for $t\in[0,1/2]$ and $D(t)=0$ for $t\in(1/2,1]$, with the convex price function
\begin{align*}
        P(z) =
        \begin{cases}
        2, &z \in [0, \frac{1}{2}] \\
        \frac{1}{1-z}, &z \in (\frac{1}{2}, 1 - \delta] \\
        \frac{1}{\delta^2}(z - 1 + \delta) + \frac{1}{\delta}, &z \in (1 - \delta,1]
        \end{cases},
\end{align*}
where $\delta < 1/2$ is small. By Claim \ref{clm:cb-average},
\begin{align*}
    \mathrm{WEL}(B_{\text{CB}}) = \tfrac{1}{2}\bigl(G(0) + G(1)\bigr) - G(\tfrac{1}{2}).
\end{align*}
One can verify that the profit-maximizing battery only shifts demand by $\delta - \delta^2$ in each period (instead of the socially optimal $1/2$). Computing the resulting PoA yields
\begin{align*}
    PoA \;\geq\; \frac{\frac{1}{2} - \ln 2\,\delta}{\frac{3}{2} - 2\delta},
\end{align*}
which diverges as $\delta \to 0^+$. The full proof is in Appendix \ref{appendix:proof-thm-counterexample}. \qed

Restricting convex functions to monomials yields a constant upper bound (Theorem~\ref{thm:PoA-mono-ub}), but we now show that relaxing to general convex polynomials does not suffice. The motivation for this construction stems from Bernstein polynomials in approximation theory~\cite{bernstein1912demonstration}. Their favorable properties for approximating convex functions, specifically uniform convergence and approaching the limit from above monotonically~\cite{aldaz2009shape}, allow us to approximate the counterexample above by convex polynomials while preserving the unboundedness of the PoA.

\begin{cor}
\label{cor:poly-lb}
    There exists at least one instance $I$, where the demand is
    \begin{align*}
        D(t)=\begin{cases}
            1, &t\in[0,\frac{1}{2}]\\
            0, &t\in(\frac{1}{2},1]
        \end{cases},
    \end{align*}
    the price function is convex polynomial, and there are no extra structural constraints, such that $PoA(I)$ is unbounded.
\end{cor}

See Appendix \ref{appendix:proof-cor:poly-lb} for the formal proof of Corollary \ref{cor:poly-lb}.

\subsection{Monomial Price Functions}

To bridge the gap between the well-behaved linear regime and the unbounded general convex regime, we analyze monomial price functions $P(z)=\alpha z^d$, where the degree $d$ directly controls the curvature. Working within $\mathcal{I}_{\text{step}}$ (Definition~\ref{def:step}), we derive upper and lower bounds on the Price of Anarchy as a function of the monomial degree.

\subsubsection{Upper Bounds}
\label{sec:ub-mono}

We first show that for any monomial degree, the efficiency loss remains bounded by a universal constant.

\begin{theorem}
    \label{thm:PoA-mono-ub}
    Define the family of instances $\mathcal{I}_{\text{mono, step}}\subset\mathcal{I}_{\text{step}}$ where all instance has a monomial price function $P(z) = \alpha z^d$ where $\alpha > 0$ and without any extra structural constraints.
    We have $PoA(\mathcal{I}_{\text{mono, step}}) \leq 2$.
\end{theorem}

Stricter bounds can be derived for specific degrees. For $d=2$, we obtain a tight characterization, and the worst-case construction provides a template for lower bounds at arbitrary degree~$d$.

\begin{theorem}
    \label{thm:PoA-mono-2-ub}
    Let $\mathcal{I}_{2\text{, step}}$ be the subset of $\mathcal{I}_{\text{mono,step}}$ where the price function is monomial of degree $2$, we have $PoA(\mathcal{I}_{2\text{, step}}) = 27/19$.
\end{theorem}

We introduce the parameterization used in the proofs. We assume without loss of generality that the 2-piecewise constant demand is $D(t)=1$ for $t\in[0,t_1]$ and $D(t)=(1-\varepsilon)x$ for $t\in(t_1,1]$. That is, we normalize the peak demand to $1$ and let the average demand over $[0,1]$ be $x \in (0,1)$. The off-peak demand is $(1-\varepsilon)x$ with $\varepsilon \in (0,1)$, giving peak duration $t_1 = \frac{x\varepsilon}{1 - x(1-\varepsilon)}$.

We assume the battery operation takes the form $B_k(t)=k(1-x)$ for $t\in[0,t_1]$ and $B_k(t)=-k\varepsilon x$ for $t\in(t_1,1]$, where $k\in[0,1]$ since $k<0$ or $k>1$ will always provide negative profit. Under this notation, the battery operation is represented by a single scalar $k$.

To simplify the analysis, let $N_1(k) = 1 - k(1-x)$ and $N_2(k) = (1-\varepsilon)x + k\varepsilon x$ denote the net demand in the peak and off-peak periods, respectively. We omit $N_1$'s and $N_2$'s dependency of $k$ when the context is clear. We can represent the revenue yields by battery operation $k$ as
\begin{align*}
    \mathrm{REV}(k)&=\int_0^1 B_k(t)\cdot P(D(t)-B_k(t))\di t=t_1\cdot k(1-x)\cdot (P(N_1)-P(N_2)).
\end{align*}

Thus, under our parameterization, we can denote $\mathcal{B}_{DCB}$, the set of optimal solutions to Programming \eqref{obj:DCB} as
\begin{align}
    \label{obj:simplied-k-poly}
    \mathcal{B}_{DCB}=\argmax_{k\in[0,1]} t_1\cdot k(1-x)\cdot (P(N_1)-P(N_2)),
\end{align}

And the social cost of the instance with parameter $x$ and $\varepsilon$ with battery operation $k$ can be represented by
\begin{align*}
    \psi_{x,\varepsilon}(k)&= \int_0^1 G(D(t)-B_k(t))\di t =\frac{\alpha}{d+1}\cdot(t_1\cdot N^{d+1}_1+(1-t_1)\cdot N^{d+1}_2).
\end{align*}

And we can represent the improvement of social cost by battery operation $k$ as $\mathrm{WEL}(k)=\psi_{x,\varepsilon}(0)-\psi_{x,\varepsilon}(k)$.

When $k=1$, the battery maximizes the improvement of social cost as stated in Claim \ref{clm:cb-average} since $N_1(1)=N_2(1)=x$ meaning the net demand equals the average demand everywhere. Therefore, $W(1)$ represents the optimal improvement on social welfare. To compare the optimal improvement with the improvement by battery operation $k$, we define $\varphi_{x,\varepsilon}(k)=\mathrm{WEL}(1)/\mathrm{WEL}(k)$, and the Price of Anarchy can be represented as $PoA=\varphi_{x,\varepsilon}(k^\star)$ where $k^\star=\argmax \varphi_{x,\varepsilon}(k)$.

\noindent \emph{Proof of Theorem \ref{thm:PoA-mono-ub}.} We start with several technical lemmas. Lemma \ref{lem:monotonicity-k} allows us to upper-bound $\varphi_{x,\varepsilon}(k^\star)$ by lower-bounding $k^\star$ using $\underline{k}$ determined in Lemma \ref{lemma:lb-k}. The proof of Lemma \ref{lem:monotonicity-k} is presented in Appendix \ref{appendix:proof-lem:monotonicity-k}.
%and \ref{appendix:proof-lemma:lb-k}.
    \begin{lemma}
        \label{lem:monotonicity-k}
        $\varphi_{x,\varepsilon}(k)$ by is monotonically decreasing with respect to $k$.
    \end{lemma}

    \begin{lemma}
        \label{lemma:lb-k}
        The optimal solution $k^*$ to Programming \eqref{obj:simplied-k-poly} is lower-bounded by $\underline{k}$ defined as
        \begin{align*}
            \underline{k} = \frac{2+d(1-x)-\sqrt{d^2(1-x)^2+4x}}{2 (d+1) (1-x)},
        \end{align*}
        where $d$ is the maximum degree of the price function $P$.
    \end{lemma}

The proof is presented in Appendix \ref{appendix:proof-lemma:lb-k}.

Next lemma ensures that $\varphi_{x,\varepsilon}(\underline{k})\leq \varphi_{x,1}(\underline{k})$, implying the upper bound is achieved when $\varepsilon=1$. The proof is postponed to Appendix \ref{appendix:proof-lem:monotonicity-eps}. We remind that $\underline{k}$ is irrelevant to $\varepsilon$.
\begin{lemma}
    \label{lem:monotonicity-eps}
    $\varphi_{x,\varepsilon}(k)$ is monotonically increasing with respect to $\varepsilon$.
\end{lemma}

The last part of our proof is to show $\varphi_{x,1}(\underline{k})\leq 2$ for any $x\in(0,1)$ and any positive integer $d$. By a change of variable $r=\underline{k}(1-x)\in(0,1/(d+1))$, we reduce this to showing $\frac{1-x^{d+1}}{1-(1-r)^{d+1}}\leq 2$, which is established via a monotonicity argument on an auxiliary function. The detailed calculation is presented in Appendix \ref{appendix:proof-completing-mono-ub}. \qed

In case $d=2$, we can obtain an analytical form of $k^\star$ and a tight estimation of $PoA=27/19$. The proof of Theorem \ref{thm:PoA-mono-2-ub} follows a different structure from the general monomial case and is presented in full in Appendix \ref{appendix:proof-thm-mono-2-ub-full}.

\subsubsection{Lower Bounds}
\label{sec:lb}

We complement the upper bounds with a family of constructions showing that the Price of Anarchy grows with the monomial degree.

\begin{theorem}
    \label{thm:PoA-mono-lb}
    Let $\mathcal{I}_{d\text{, step}}$ be the subset of $\mathcal{I}_{\text{mono,step}}$ where the price function is monomial of degree $d$, we have
    $PoA(\mathcal{I}_{d\text{, step}})\geq 1/(1-(\frac{d}{d+1})^{d+1})$. When the degree $d$ approaches infinity, $PoA$ is lower-bounded by $e/(e-1)$.
\end{theorem}

Observing the tightness in the linear case ($d=1$, $PoA=4/3$) and the quadratic case ($d=2$, $PoA=27/19$), and supported by extensive numerical verification, we conjecture that the lower bound in Theorem~\ref{thm:PoA-mono-lb} is tight for every positive integer degree~$d$, suggesting a direct, monotonic relationship between the degree and market inefficiency.

\noindent \emph{Proof sketch of Theorem \ref{thm:PoA-mono-lb}.} Consider demand $D(t)=1$ for $t\in[0,\varepsilon]$, $D(t)=0$ for $t\in(\varepsilon,1]$, and price $P(z)=(d+1) z^d$ with $G(z)=z^{d+1}$. By Claim \ref{clm:cb-average}, $\mathrm{WEL}(B_{CB})=\varepsilon-\varepsilon^{d+1}$. Upper-bounding the optimal decentralized operation via first-order conditions gives $x^\star\leq 1/(d+1)$, which yields
\begin{align*}
    PoA \;\geq\; \frac{1-\varepsilon^d}{1-\bigl(\frac{d}{d+1}\bigr)^{d+1}}.
\end{align*}
Taking $\varepsilon\to0$ yields the result. The full proof is in Appendix \ref{appendix:proof-thm-mono-lb}. \qed
    
\section{Extension to Multiple Batteries}
\label{sec:multi-battery}

The preceding sections quantify the efficiency loss of a single strategic battery. A classical remedy for market power is competition: in many economic settings, increasing the number of strategic agents drives the equilibrium toward social optimality. In this section, we extend the linear pricing analysis to $n$ identical batteries competing simultaneously. We prove that the resulting game admits a unique pure Nash equilibrium, which is necessarily symmetric, and that competition monotonically reduces the Price of Anarchy, converging to perfect efficiency as $n\to\infty$.

\subsection{Setup and Nash Equilibrium}

Consider $n$ identical batteries operating simultaneously in a market with linear pricing $P(z)=az+b$. Each battery $i\in\{1,\ldots,n\}$ chooses an operation $B_i\in\mathcal{B}\cap\Omega$. The aggregate battery operation is $S(t)=\sum_{i=1}^n B_i(t)$, and the net demand becomes $N(t)=D(t)-S(t)$. The box constraint \eqref{constraint:Box} is imposed on the aggregate net demand, while the non-anticipating and periodicity requirements and the operational set $\Omega$ apply to each $B_i$ individually. Under linear pricing the price function is defined for all net-demand levels, so this aggregate feasibility requirement does not affect the equilibrium or welfare computations below.

Each battery is a strategic agent that observes the strategies of all other batteries and chooses its own operation to maximize its expected revenue. Specifically, the revenue of battery $i$ depends on both its own operation and the aggregate:
\begin{align*}
    \mathrm{REV}_i(B_i; B_{-i}) = \mathbb{E}\left[\int_0^1 B_i(t)\cdot P\bigl(D(t)-S(t)\bigr)\,\di t\right],
\end{align*}
where $B_{-i}=(B_1,\ldots,B_{i-1},B_{i+1},\ldots,B_n)$ denotes the strategies of all batteries other than $i$. The social welfare is measured by the total reduction in generation cost:
\begin{align*}
    \mathrm{WEL}(B_1,\ldots,B_n) = \mathbb{E}\left[\int_0^1 G(D(t))-G(D(t)-S(t))\,\di t\right].
\end{align*}

We adopt the standard game-theoretic solution concept of Nash equilibrium, where no battery can unilaterally improve its revenue.

\begin{definition}[Nash Equilibrium]
\label{def:nash}
A strategy profile $(B_1^*,\ldots,B_n^*)$ with each $B_i^*\in\mathcal{B}\cap\Omega$ is a \textit{(pure) Nash equilibrium} if no battery can increase its revenue by unilaterally deviating:
\begin{align*}
    \mathrm{REV}_i(B_i^*; B_{-i}^*) \geq \mathrm{REV}_i(B_i; B_{-i}^*), \quad \forall B_i\in\mathcal{B}\cap\Omega,\quad \forall i\in\{1,\ldots,n\}.
\end{align*}
Equivalently, each $B_i^*$ is a best response to the other batteries' strategies $B_{-i}^*$.
\end{definition}

The Price of Anarchy for $n$ batteries is the ratio of the centralized welfare (where a planner jointly controls all $n$ batteries) to the welfare at Nash equilibrium. When $n=1$, this reduces to the single-battery PoA defined earlier.

\subsection{Efficiency at Equilibrium}

Under linear pricing, the game has a particularly clean structure. Using the inner product notation from Section~\ref{sec:linear}, the welfare and revenue simplify to $\mathrm{WEL}=\tfrac{1}{2}a(2\langle D,S\rangle - \|S\|^2)$ and $\mathrm{REV}_i=a(\langle D,B_i\rangle - \langle S,B_i\rangle)$, extending Equations \eqref{eq:W-linear}--\eqref{eq:REV-linear}.

A key structural insight is that this game admits an exact potential function:
\begin{align*}
    \Phi(B_1,\ldots,B_n) = \langle D,S\rangle - \tfrac{1}{2}\|S\|^2 - \tfrac{1}{2}\sum_{i=1}^n\|B_i\|^2.
\end{align*}
One can verify that $\nabla_{B_i}\Phi = D - S - B_i = \frac{1}{a}\nabla_{B_i}\mathrm{REV}_i$, confirming that unilateral deviations in revenue are exactly captured by changes in $\Phi$. The quadratic form of $\Phi$ corresponds to $-\frac{1}{2}\mathbf{B}^\top(J+I)\mathbf{B}$, where $J$ is the $n\times n$ all-ones matrix and $I$ the identity. Since $J+I$ is positive definite (with eigenvalues $n+1$ and $1$ of multiplicities $1$ and $n-1$, respectively), $\Phi$ is strictly concave on $(\mathcal{B}\cap\Omega)^n$.

\begin{theorem}
\label{thm:PoA-linear-n}
    For $n$ identical batteries under linear pricing $P(z)=a\cdot z+b$ with $a>0,b\geq 0$, the game admits a unique pure Nash equilibrium, which is symmetric: $B_i^*=B^*$ for all $i$. The Price of Anarchy at this equilibrium satisfies $PoA(\mathcal{I}_{\mathrm{lin}})=(n+1)^2/\bigl(n(n+2)\bigr)$. In particular, $PoA=4/3$ when $n=1$ and $PoA\to 1$ as $n\to\infty$.
\end{theorem}

\noindent \emph{Proof sketch.} Existence and uniqueness of the Nash equilibrium follow from the strict concavity of $\Phi$: any Nash equilibrium must maximize $\Phi$ over the convex set $(\mathcal{B}\cap\Omega)^n$, and a strictly concave function admits at most one maximizer. Symmetry follows from the permutation invariance of the game. For the PoA bound, the proof follows the same variational inequality strategy as Theorem \ref{thm:PoA-linear}. At the symmetric equilibrium with $S^*=nB^*$, each battery's first-order condition becomes
\begin{align*}
    \langle D-(n+1)B^*,\, B-B^*\rangle \leq 0, \quad \forall\, B\in\mathcal{B}\cap\Omega.
\end{align*}
Substituting $B=B_{CB}$ and $B=0$ and combining with $\alpha=(n+1)^2/(n(n+2))$ yields
\begin{align*}
    \alpha\bigl(2\langle D,B^*\rangle - n\|B^*\|^2\bigr) - \bigl(2\langle D,B_{CB}\rangle - n\|B_{CB}\|^2\bigr) \;\geq\; n\left\|\tfrac{n+1}{n}B^*-B_{CB}\right\|^2 \geq 0,
\end{align*}
establishing the upper bound. For tightness, consider the $\mathcal{I}_{\text{step}}$ instance with $D(t)=1$ for $t\in[0,1/2]$, $D(t)=0$ for $t\in(1/2,1]$, and $P(z)=az$. Writing $\widetilde{D}=D-\widebar{D}$ for the mean-zero fluctuation, the Nash equilibrium is $B^*=\widetilde{D}/(n+1)$ and the centralized optimum is $B_{CB}=\widetilde{D}/n$, achieving the ratio $(n+1)^2/(n(n+2))$ exactly. The full proof is in Appendix \ref{appendix:proof-thm-linear-n}. \qed

\section{Discussion and Future Research Directions}

We have established that the structure of the price function is the main determinant of market efficiency for battery storage. For linear prices, a variational inequality argument yields a tight PoA of $4/3$ for a single battery, and a potential-game argument gives $(n+1)^2/(n(n+2))$ for $n$ batteries at the unique Nash equilibrium, converging to $1$ as $n$ grows. As the number of competing batteries increases, each battery's market power diminishes and the equilibrium converges to the social optimum, providing a theoretical justification for encouraging competition among storage operators. For general convex prices the loss can be unbounded, even for convex polynomials, so linearity is a sharp boundary. For monomials of degree $d$ the loss stays bounded by $2$ while growing with $d$, with a tight value of $27/19$ at $d=2$ and a lower bound approaching $e/(e-1)$ as $d\to\infty$. Open directions include extending the multi-battery analysis to non-linear pricing, tight bounds for arbitrary degree~$d$, transmission losses, and continuous-time stochastic demand.

\section*{Acknowledgements}
The work of Zhile Jiang and Stratis Skoulakis was funded by the Villum Young Investigator Award no.~72091.

\bibliographystyle{plain}
\bibliography{reference}

\appendix

\section{Proof of Claim \ref{clm:cb-average}}
\label{appendix:proof-clm:cb-average}
By Constraints \eqref{constraint:Periodicity}, we have
    \begin{align*}
        \int_0^1 D(t)-B(t)\di t=\int_0^1 D(t)\di t=\widebar{D}.
    \end{align*}
    Recall the generation cost is always a convex function. Jensen's inequality ensures 
    \begin{align*}
        \int_0^1 G(D(t)-B(t))\di t\geq G\left(\int_0^1 D(t)-B(t)\di t\right)= G(\widebar{D}).
    \end{align*}
    Thus, the maximum improvement $\max_{B\in\mathcal{B}} \mathrm{WEL}(B)$is upper-bounded as 
    \begin{align*}
       \max_{B\in\mathcal{B}} \mathrm{WEL}(B)\leq \int_0^1 G(D)\di t-G(\widebar{D}).
    \end{align*}
    By setting $B_{CB}(t)=D(t)-\widebar{D}$, which is in $\mathcal{B}$, we can achieve this improvement. 
\qed

\section{Full Proof of Theorem \ref{thm:PoA-counterexample}}
\label{appendix:proof-thm-counterexample}

Consider a deterministic demand
\begin{align*}
    D(t)=\begin{cases}
        1,&t\in[0,1/2]\\
        0,&t\in(1/2,1]
    \end{cases}
\end{align*}
and the price function from the main text. The generation cost function is
\begin{align*}
        G(z) =
        \begin{cases}
        2z & z \in [0, \frac{1}{2}] \\
        -\ln{(1 - z)} + \ln{\frac{1}{2}} + 1 & z \in (\frac{1}{2}, 1 - \delta] \\
        \frac{(z - 1 + \delta)^2}{2\delta^2} + \frac{z - 1 + \delta}{\delta} -\ln{2\delta} + 1 &z \in (1 - \delta,1]
        \end{cases}.
    \end{align*}

Both $G$ and $P$ are convex on $[0,1]$. By Claim \ref{clm:cb-average}, $\text{WEL}(B_{\text{CB}}) = \frac{1}{2}(G(0) + G(1)) - G(\frac{1}{2})$.

Under $\mathcal{I}_{step}$, Programming \eqref{obj:DCB} is equivalent to $\max_{x \in [0,1]} \frac{1}{2}x(P(1-x) - P(x))$. By monotonicity of $P$, it suffices to consider $x \in [0,\frac{1}{2}]$. If $x \in [0, \delta)$:
$x(P(1-x) - P(x)) = x(-\frac{x}{\delta^2} + \frac{2}{\delta} - 2)$,
with maximum $(1-\delta)^2$ at $x = \delta-\delta^2$. If $x \in [\delta, \frac{1}{2}]$:
$x(P(1-x) - P(x)) = 1 - 2x$,
with maximum $1-2\delta$ at $x = \delta$. Since $1-\delta^2>1-2 \delta$, we get
$$
\mathrm{WEL}(B_{\text{DCB}}) = \tfrac{1}{2}(G(0) + G(1)) - \tfrac{1}{2}(G(1-\delta+\delta^2) + G(\delta-\delta^2)).
$$
Then
\begin{align*}
PoA &= \frac{\text{WEL}(B_{\text{CB}})}{\text{WEL}(B_{\text{DCB}})} \geq \frac{G(1) - 2G(\frac{1}{2})}{G(1) - G(1- \delta) - G(\delta)} = \frac{\frac{1}{2} - \ln{2\delta}}{\frac{3}{2} - 2\delta}.
\end{align*}
The inequality is ensured by the convexity of $G$. The limit $\lim_{\delta \rightarrow 0^+} PoA = \infty$ completes the proof. \qed

\section{Proof of Corollary \ref{cor:poly-lb}}
\label{appendix:proof-cor:poly-lb}

Our proof relies on the Bernstein Operator defined below.

\begin{definition}(Bernstein Operator \cite{bernstein1912demonstration})
    A Bernstein operator $\mathscr{B}_n: C[0,1] \rightarrow \Pi_n$, defined by 
    $$\mathscr{B}_nf(x) = \sum_{k=0}^nf(\frac{k}{n})\binom{n}{k}x^k(1-x)^{n-k},$$
    where $C[0,1]$ denotes the space of continuous functions defined on the interval $[0,1]$ and $\Pi_n$ denotes the space of polynomials of degree at most $n$.
\end{definition}

We note the following readily established properties for future reference which can be viewed as a special case of the generalized result established in \cite{aldaz2009shape}.\begin{itemize}
    \item \emph{Uniform Convergence:} For any function $f \in C[0,1]$ and corresponding $\{\mathscr{B}_nf\}_{n = 1}^{\infty}$, for all $\varepsilon > 0$, there exists some integer $N$ such that $\forall n > N$ and $\forall x \in [0,1]$, we have $\left| \mathscr{B}_nf(x) - f(x)\right| < \varepsilon$:

    \item \emph{Convexity Preserving:} If function $f \in C[0,1]$ is convex, then $\mathscr{B}_nf$ is a convex polynomial for any $n$.
    \item \emph{Above Approaching:} For any convex function $f \in C[0,1]$ and corresponding $\{\mathscr{B}_nf\}_{n = 1}^{\infty}$, the inequalities $\mathscr{B}_{n}f(x) \geq \mathscr{B}_{n+1}f(x) \geq f(x)$ always hold for $x \in [0,1]$. 
\end{itemize}

By uniform convergence, we select a concrete $\varepsilon$ s.t. $0 < \varepsilon \leq \delta$ and then a large enough degree $n$. Consider the convex price function $P$ used in above counterexample and its corresponding polynomial $\mathscr{B}_nP$. By definition, we get two cost functions $G, \mathscr{G}_n$ and gap between this two functions for all $x \in [0,1]$:
$$
0 \leq \mathscr{G}_n(x) - G(x) = \int_0^z \mathscr{B}_nP(x) \di t - \int_0^z P(x) \di t = \int_0^z \mathscr{B}_nP(x) - P(x) \di t \leq \int_0^z \varepsilon \di t = \varepsilon z \leq \varepsilon.
$$
Denote ${I}_n = (D, \mathscr{B}_nP, \Omega)\in \mathcal{I}_{\text{step}}$. Similar to above analysis, applying Claim \ref{clm:cb-average} on the convex polynomial $\mathscr{B}_nP$ can get the maximum improvement 
$$\text{WEL}_{I_n}(B_{\text{CB}}(I_n)) = \frac{1}{2}(\mathscr{G}_n(0) + \mathscr{G}_n(1)) - \mathscr{G}_n(\frac{1}{2}) \geq \frac{1}{2}(G(0) + G(1)) - G(\frac{1}{2}) - \varepsilon.$$

Use function $F$ to denote the profit function $x(P(1-x) - P(x))$, ignoring the constant $\frac{1}{2}$ for simplicity and $\mathscr{F}$ to denote the function $x(\mathscr{B}_nP(1-x) - \mathscr{B}_nP(x))$. Next is to show $x^* = \argmax_{x\in[0,1]} \mathscr{F}(x) \in [0, 2\delta]$ by contradiction.

Bound the function value using above approaching and uniform convergence
$$x(P(1-x) - P(x)) - \varepsilon x \leq \mathscr{F}(x) = x(\mathscr{B}_nP(1-x) - \mathscr{B}_nP(x)) \leq x(P(1-x) - P(x)) + \varepsilon x.$$
In other words, $\mathscr{F}(x) \in [F(x) - \varepsilon x, F(x) + \varepsilon x]$. Thus $\mathscr{F}(\delta) \in [1- 2\delta - \varepsilon\delta, 1 - 2\delta + \varepsilon\delta]$. For all $x > 2\delta$, $$\mathscr{F}(x) \leq F(x) + \varepsilon x \leq F(x) + \varepsilon.$$
By monotonicity of $P$, it suffices to consider $x \in (2\delta, \frac{1}{2}]$. Together with the above analysis, $F(x)$ is monotonically decreasing when $x \in (2\delta, \frac{1}{2}]$. Thus, 
$$
\forall x \in (2\delta,1], \mathscr{F}(x) \leq F(2\delta) + \varepsilon = 1 - 4\delta + \varepsilon \leq 1 - 3\delta < 1 - 2\delta - \varepsilon \delta \leq \mathscr{F}(\delta).
$$

Bound the $W_{I_n}(B_{DCB}(I_n))$ as follows. Due to convexity of $G$, $G(1-x) + G(x)$ is decreasing when $x \leq \frac{1}{2}$. Then,
\begin{align*}
\text{WEL}_I(B_{\text{DCB}}(I)) &= \frac{1}{2}(\mathscr{G}(0) + \mathscr{G}(1)) - \frac{1}{2}(\mathscr{G}(1-x^*) + \mathscr{G}(x^*)) \\
&\leq \frac{1}{2}(G(0) + G(1)) + \varepsilon - \frac{1}{2}(G(1-x^*) + G(x^*)) \\
&\leq \frac{1}{2}(G(0) + G(1)) + \varepsilon - \frac{1}{2}(G(1-2\delta) + G(2\delta))
\end{align*}

We conclude the Price of Anarchy is unbounded by
\begin{align*}
PoA(I_n) &= \frac{\text{WEL}_{I_n}(B_{\text{CB}}(I_n))}{\text{WEL}_{I_n}(B_{\text{DCB}}(I_n))} \\
& \geq \frac{\frac{1}{2}(G(0) + G(1)) - G(\frac{1}{2}) - \varepsilon}{\frac{1}{2}(G(0) + G(1)) + \varepsilon - \frac{1}{2}(G(1-2\delta) + G(2\delta))} \\
& \geq \frac{\frac{1}{2}(G(0) + G(1)) - G(\frac{1}{2}) - \delta}{\frac{1}{2}(G(0) + G(1)) + \delta - \frac{1}{2}(G(1-2\delta) + G(2\delta))} \\
& = \frac{G(1) - 2G(\frac{1}{2}) -2\delta}{G(1) - G(1- 2\delta) - G(2\delta) + 2\delta} \\
& = \frac{\frac{1}{2} - \ln{2\delta} - 2\delta}{\frac{3}{2} + \ln{2} - 2\delta},
\end{align*}
together with $\lim_{\delta \rightarrow 0^+} PoA = \infty$.
\qed

\section{Proof of Lemma \ref{lem:monotonicity-k}}
\label{appendix:proof-lem:monotonicity-k}
%It suffices to show that $\psi_{x,\varepsilon}(k)$ is monotonically decreasing. 
We show that $\varphi_{x,\varepsilon}(k)$ is monotonically decreasing.
Take the derivative, we have
    \begin{align*}
        \frac{\partial \varphi_{x,\varepsilon}(k)}{\partial k}=-\frac{\alpha \varepsilon x (1-x)}{1-(1-\varepsilon)x} \left(N_1^{d} - N_2^{d} \right).
    \end{align*}
    The function is monotonically decreasing if $N_1(k) \geq N_2(k)$, which is true since
    \begin{align*}
    N_1 - N_2=1-k(1-x)-x(1-(1-k)\varepsilon)=(1-k) ((1-x) + x\varepsilon)\geq 0.
    \end{align*}  
\qed

\section{Proof of Lemma \ref{lemma:lb-k}}
\label{appendix:proof-lemma:lb-k}

    Define $Z(k)=t_1\cdot k(1-x)\cdot (P(N_1)-P(N_2))$, we can simplify it as

\begin{align*}
    Z(k)&=\frac{(1-x)x\varepsilon}{1-(1-\varepsilon)x}\cdot k\cdot\alpha \cdot(N^d_1-N^d_2)\\
    &=\frac{\alpha(1-x)x\varepsilon}{1-(1-\varepsilon)x}\cdot k\cdot(N_1-N_2)\cdot X=\alpha(1-x)x\varepsilon \cdot k(1-k)\cdot X,
\end{align*}
where
\begin{align*}
    X=\sum_{a+b=d-1,\atop a\geq0,b\geq0}N_1^a N_2^b.
\end{align*}
Since $x$,$\varepsilon$, and $\alpha$ are all constant fixed by the instance, it is sufficient to maximize
\begin{align*}
    f_{x,\varepsilon}(k)=k(1-k)\cdot X.
\end{align*}
Taking the derivative of $f(k)$, we have
\begin{align*}
    f'(k)=(1-2k)\cdot X+k(1-k)\cdot X'.
\end{align*}
And we can lower-bound $X'$ as follows.
\begin{align*}
    X'&=\sum_{a+b=d-1,\atop a\geq0,b\geq0}(aN'_1 N_2+bN_1 N'_2)\cdot N_1^{a-1}N_2^{b-1}\\
    &\geq\sum_{a+b=d-1,\atop a\geq0,b\geq0}(d-1)(N'_1 N_2)\cdot N_1^{a-1}N_2^{b-1}=(d-1)\frac{N_1'}{N_1}\cdot X,
\end{align*}

The inequality holds due to $N_1'N_2-N_1 N_2'=(x-1)x(1-(1-k)\varepsilon-(1-k(1-x)))x\varepsilon=x(x(1-\varepsilon)-1)<0$.
Thus,
\begin{align*}
    f'(k)\geq\left(1-2k+k(1-k)(d-1)\frac{N_1'}{N_1}\right)\cdot X.
\end{align*}

Since $f'_{x,\varepsilon}(0) > 0$ and $f'_{x,\varepsilon}(1) < 0$,  $k^\star \in [k_0, 1]$ when function $f_{x, \varepsilon}(k^\star)$ attains maximum where $k_0$ is the least root of $f'_{x,\varepsilon}(k) = 0$. We know $g(k) = (1-2k+k(1-k)(d-1)\frac{N_1'}{N_1})X = 0$ has a root in the interval $[0, k_0]$, namely $\underline{k}$, because $g(0) > 0$ and $g(k_0) \leq f'(k_0) = 0$. Thus, the optimal solution $k^*$ to Programming \eqref{obj:simplied-k-poly} is lower-bounded by $\underline{k}$.

By noticing $X > 0$, it's sufficient to solve $1 - 2k + k(1-k)\cdot(d - 1)\cdot(N_1'/N_1) = 0$ which is equivalent to the following quadratic equation
\begin{align*}
    (1-x)(d+1) k^2 - (2+d(1-x))k + 1 = 0.
\end{align*}
The two roots are
\begin{align*}
    k_1 &=\frac{2+d(1-x)-\sqrt{d^2(1-x)^2+4x}}{2 (d+1) (1-x)},\\
    k_2 &=\frac{2+d(1-x)+\sqrt{d^2(1-x)^2+4x}}{2 (d+1) (1-x)} > 1.
\end{align*}
By discarding the root $k_2$, we can conclude our lemma.
\qed

\section{Proof of Lemma \ref{lem:monotonicity-eps}}
\label{appendix:proof-lem:monotonicity-eps}
Showing $\varphi_{x,\varepsilon}(k)$ is always monotonically increasing with respect to $\varepsilon$ is equivalent to prove its logarithmic derivative with respect to $\varepsilon$, i.e.,
    \begin{align*}
        \frac{\partial \ln \varphi_{x,\varepsilon}(k)}{\partial \varepsilon}&=\frac{\partial \ln \mathrm{WEL}(1)}{\partial \varepsilon}-\frac{\partial \ln \mathrm{WEL}(k)}{\partial \varepsilon}\\
        &=\left(\frac{\partial\mathrm{WEL}(1)}{\partial \varepsilon}\right)/{\mathrm{WEL}(1)}-\left(\frac{\partial\mathrm{WEL}(k)}{\partial \varepsilon}\right)/{\mathrm{WEL}(k)},
    \end{align*}
    is always non-negative. Let $R(k)=\left(\frac{\partial\mathrm{WEL}(k)}{\partial \varepsilon}\right)/\mathrm{WEL}(k)$. It is sufficient to prove that $R(k)$ is monotonically increasing with respect to $k$ on $[0,1]$. Notice that $\partial\mathrm{WEL}(0)/\partial \varepsilon=0$ and $\mathrm{WEL}(0)=0$ by the definition of $\mathrm{WEL}(k)$. And since $\partial\mathrm{WEL}(k)/\partial \varepsilon$ and $\mathrm{WEL}(k)$ are both continuous functions on $[0,1]$ and differentiable on $(0,1)$, we can use the following lemma to prove the monotonicity of $R(k)$.
    \begin{lemma}[Lemma 2.2 in \cite{AVV93}]
        Let $f,g$ be continuous functions defined in $[a,b]$ and differentiable in $(a, b)$. Suppose that $f(a)=g(a)=0$ and $g'(x)\neq 0$ for all $x\in(a,b)$. If $f'/g'$ is increasing on $[a,b]$ then so is $f/g$.
    \end{lemma}
    
    Thus, it suffices to prove $Q(k)$ defined below is increasing.
    \begin{align*}
        Q(k)=\frac{\partial^2\mathrm{WEL}(k)}{\partial k\partial\varepsilon}/\frac{\partial\mathrm{WEL}(k)}{\partial k}=\frac{\partial^2\psi(k)}{\partial k\partial\varepsilon}/\frac{\partial\psi(k)}{\partial k}=\frac{\partial \ln \psi'_{x,\varepsilon}(k) }{\partial \varepsilon}.
    \end{align*}

    We have that
    \begin{align*}
        \ln \psi'_{x,\varepsilon}(k)&=\ln \left(\frac{\alpha\varepsilon x(1-x)}{1-(1-\varepsilon)x}\cdot\left(N_1^{d}-N_2^{d}\right)\right)\\
        &=\ln\left(\alpha x(1-x)\right)-\ln\left(1-(1-\varepsilon)x\right)+\ln\varepsilon+\ln\left(N_1^{d}-N_2^{d}\right).
    \end{align*}
    Recall $N_1=1-k(1-x)$ and $N_2=x(1-(1-k)\varepsilon)$. Thus,
    \begin{align*}
        Q(k)&=\frac{\partial \ln \psi'_{x,\varepsilon}(k) }{\partial \varepsilon}\\
        &=-\frac{x}{1-(1-\varepsilon)x}+\frac{1}{\varepsilon}+\frac{d\cdot x\cdot (1-k)\cdot N_2^{d-1}}{N_1^{d}-N_2^{d}}.
    \end{align*}
    We observe that only the third term contains $k$. Define  
    \begin{align*}
        P(k)=\frac{dx(1-k)N_2^{d-1}}{N_1^{d}-N_2^{d}}.
    \end{align*}
    It suffices to prove $P(k)$ is increasing. Let $z=N_1/N_2$, we can rewrite $P(k)$ as 
    \begin{align*}
        P(k)=\frac{d\cdot x}{(1-x)+x\varepsilon}\cdot\frac{(N_1-N_2)\cdot N_2^{d-1}}{N_1^{d}-N_2^{d}}=\frac{d\cdot x}{(1-x)+x\varepsilon}\cdot\frac{z-1}{z^{d}-1}=\frac{d\cdot x}{(1-x)+x\varepsilon}\cdot\frac{1}{\sum_{i=0}^{d-1}z^i}.
    \end{align*}
    It is not hard to see $P(k)$ is strictly decreasing with respect to $z$. Together with the fact that $z=N_1/N_2=(1-k(1-x))/(x-\varepsilon x(1-k))$ is strictly decreasing with respect to $k$. $P(k)$ is increasing with respect $k$.
    
\section{Completing the Proof of Theorem \ref{thm:PoA-mono-ub}}
\label{appendix:proof-completing-mono-ub}

We need to show $\varphi_{x,1}(\underline{k})\leq 2$ for any $x\in(0,1)$ and any positive integer $d$. We can bound it as follows.

    \begin{align*}
        \varphi_{x,1}(\underline{k})&=\frac{x-x^{d+1}}{x-(x(1-\underline{k}(1-x))^{d+1}+(1-x)(\underline{k}x)^{d+1})}\\
        &\leq \frac{x-x^{d+1}}{x-(x(1-\underline{k}(1-x))^{d+1}+(1-x)(x^{d+1}\frac{1-(1-\underline{k}(1-x))^{d+1}}{1-x^{d+1}})}\\
        &=\frac{1-x^{d+1}}{(1-(1-\underline{k}(1-x))^{d+1})}
    \end{align*}

    The inequality holds since $1-k^{d+1}(1-x)^{d+1}\geq(1-k(1-x))^{d+1}$ for any $k \in [0,1]$. Let
    \begin{align*}
        r=\underline{k}(1-x)=\frac{2+d(1-x)-\sqrt{d^2(1-x)^2+4x}}{2(d+1)}.
    \end{align*}

    Notice that $r\in(0,1/(d+1))$ when $x\in(0,1)$. Solving $(2(d+1) r - 2 - d(1-x))^2 = d^2 (1-x)^2 + 4x$ gives $x = \frac{(1 - r)(1 - (d+1) r)}{1 - dr}$.

    To prove $\varphi_{x,1}(\underline{k})\leq 2$, it suffices to show
    $\frac{1-x^{d+1}}{1-(1-r)^{d+1}}\leq 2$
    for $r\in(0,1/(d+1))$, which is equivalent to
    \begin{align*}
        f(r)=(\frac{1-(d+1) r}{1-dr})^{d+1}+\frac{1}{(1-r)^{d+1}}\geq 2.
    \end{align*}

    Since $f(0)=2$, it suffices to show $f'(r) \geq 0$,
    which reduces to showing $h(r)=(d+2) \ln(1 - dr) - (d+2) \ln(1 - r) - d \ln(1 - (d+1) r)>0$ for $r\in(0,1/(d+1))$.
    Since $h(0)=0$ and
    \begin{align*}
        h'(r)=\frac{d(d+2)r^2}{(1-r)(1-d r)(1-(d+1)r)}>0,
    \end{align*}
    the result follows.
\qed

\section{Full Proof of Theorem \ref{thm:PoA-mono-2-ub}}
\label{appendix:proof-thm-mono-2-ub-full}

We show that the bound $27/19$ is attained as $x\to 0$ and $\varepsilon \to 0$, while the bound $2$ of Theorem \ref{thm:PoA-mono-ub} is attained as $x\to1$ and $\varepsilon=1$. Throughout, the price function is $P(z)=\alpha z^2$. The proof has three steps. Step 1 derives the analytic form of the optimal operation $k^\star$ and its monotonicity in $x$ and $\varepsilon$ (Lemma \ref{lem:opt-2-mono}). Step 2 expresses the PoA purely in terms of $k^\star$. Step 3 shows this expression is monotone in $k^\star$ (Lemma \ref{lem:monotone-k-2}). Combining the three steps locates the maximizing instance.

\paragraph{Step 1: optimal operation $k^\star$.}

\begin{lemma}
\label{lem:opt-2-mono}
    The optimization problem $\max_{k\in[0,1]} Z(k)=t_1\cdot k(1-x)\cdot (P(N_1)-P(N_2))$ has a unique solution
    \begin{align*}
    k^\star=\begin{cases}
        \frac{2(1-\varepsilon x)-\sqrt{(\varepsilon^2+3)x^2 - 2 \varepsilon x+ 1 }}{3(1-x-\varepsilon x)}&\quad \text{if } \varepsilon\neq (1-x)/x\\
        1/2&\quad \text{otherwise}.
    \end{cases}
    \end{align*}
    And $k^\star$ is monotonically increasing with respect to $x$ and $\varepsilon$.
\end{lemma}

\noindent\emph{Proof of Lemma \ref{lem:opt-2-mono}.}
 We rewrite $Z(k)$ as
    \begin{align*}
        Z(k)&=3k(1-x)\cdot t_1\cdot \left(\alpha(1-k(1-x))^2-\alpha((1-\varepsilon)x+k\cdot\varepsilon x)^2\right)\\
    &=3\varepsilon x(1-x)\alpha \cdot k(1-k)(1+x-\varepsilon x-k(1-x-\varepsilon x)).
    \end{align*}
    Since $x$,$\varepsilon$, and $\alpha$ are constants fixed by the instance, it is sufficient to maximize $f_{x,\varepsilon}(k)=k(1-k)(1+x-\varepsilon x-k(1-x-\varepsilon x))$. When $\varepsilon= (1-x)/x$, $f_{x,\varepsilon}(k)=2x\cdot k(1-k)$, which is maximized when $k=1/2$. Otherwise, take the derivative of $f_{x,\varepsilon}(k)$, we have 
    \begin{align*}
    f'_{x,\varepsilon}(k)&=(1-k)(1+x-\varepsilon x-k(1-x-\varepsilon x))-k(1+x-\varepsilon x-k(1-x-\varepsilon x))-k(1-k)(1-x-\varepsilon x)\\
    &=(1-2k)(1+x-\varepsilon x-k(1-x-\varepsilon x))-k(1-k)(1-x-\varepsilon x)\\
    &=3(1-x-\varepsilon x)k^2-4(1-\varepsilon x)k+1+x-\varepsilon x.
    \end{align*}
    
The optimal can be determined by solving $f'_{x,\varepsilon}(k)=0$ for $k$, which is a quadratic equation. The roots of this equation are
\begin{align*}
    k_1&=\frac{2(1-\varepsilon x)+\sqrt{(\varepsilon^2+3)x^2 - 2 \varepsilon x+ 1 }}{3(1-x-\varepsilon x)}\\
    k_2&=\frac{2(1-\varepsilon x)-\sqrt{(\varepsilon^2+3)x^2 - 2 \varepsilon x+ 1 }}{3(1-x-\varepsilon x)}
\end{align*}

Notice that $f'_{x,\varepsilon}(0)=1+x-\varepsilon x > 0$ and $f'_{x,\varepsilon}(1)=-2x < 0$. Together with the fact that $f'_{x,\varepsilon}$ is a continuous function on $[0,1]$, $f_{x,\varepsilon}$ must have a maximum on $[0,1]$. And this maximum must be one of the two roots of $f'_{x,\varepsilon}(k)=0$. We will then prove $k_2$ is our desired root in both cases $\varepsilon>(1-x)/x$ and $\varepsilon<(1-x)/x$. 
    
When $\varepsilon>(1-x)/x$,  we notice $k_1<0$ and $k_2>0$. So, $k_2$ must be the root in $[0,1]$. 
When $\varepsilon<(1-x)/x$, we have that $\lim_{k\to\infty}f'_{x,\varepsilon}(k)=\infty$. Together with the fact $f'_{x,\varepsilon}(1)<0$, there must be a root that is greater than $1$. By noticing $3(1-x-\varepsilon x)>0$, we identify $k_2$, the smaller root, lies between $0$ and $1$.

Thus, we have
\begin{align*}
    k^*=\begin{cases}
        k_2&\quad \text{if } \varepsilon\neq (1-x)/x\\
        1/2&\quad \text{otherwise}.
    \end{cases}
\end{align*}

We prove $k^*$ is monotonically increasing with respect to $x$ and $\varepsilon$ by showing $\partial k_2/\partial x\geq0$ and $\partial k_2/\partial \varepsilon\geq0$. Define $P=2(1-\varepsilon x)-\sqrt{(\varepsilon^2+3)x^2 - 2 \varepsilon x+ 1 }$ and $Q=3(1-x-\varepsilon x)$. It is sufficient to prove that $g(\varepsilon,x)=Q\frac{\partial P}{\partial x}-P\frac{\partial Q}{\partial x}\geq 0$ and $h(\varepsilon,x)=Q\frac{\partial P}{\partial \varepsilon}-P\frac{\partial Q}{\partial \varepsilon}\geq 0$.
    Let $S=(\varepsilon^2+3)x^2-2\varepsilon x+1$, we have that
    \begin{align*}
        g(\varepsilon,x)&=Q\frac{\partial P}{\partial x}-P\frac{\partial Q}{\partial x}\\
        &=Q\cdot \left(-2\varepsilon - \frac{(\varepsilon^2+3)x - \varepsilon}{\sqrt{S}}\right)-P\cdot(-3(1+\varepsilon))\\
        &= 3(1-x-\varepsilon x) \cdot \left( \frac{-2\varepsilon\sqrt{S} - ((\varepsilon^2+3)x - \varepsilon)}{\sqrt{S}} \right)- \left( 2(1-\varepsilon x) - \sqrt{S} \right) \cdot (-3(1+\varepsilon))\\
        &= \frac{3}{\sqrt{S}} \left((1-x-\varepsilon x) (-2\varepsilon\sqrt{S} - (\varepsilon^2+3)x + \varepsilon) + (1+\varepsilon)\sqrt{S} ( 2(1-\varepsilon x) - \sqrt{S} ) \right)\\
        &=\frac{3\cdot(2 \sqrt{S}- ((3 - \varepsilon) x + 1))}{\sqrt{S}},
    \end{align*}
    and
    \begin{align*}
        h(\varepsilon,x)&=Q\frac{\partial P}{\partial \varepsilon}-P\frac{\partial Q}{\partial \varepsilon}\\
        &=Q\cdot (-2x - \frac{\varepsilon x^2 - x}{\sqrt{S}})-P\cdot(-3x)\\
        &=3(1-x-\varepsilon x) \cdot \left(-2x - \frac{\varepsilon x^2 - x}{\sqrt{S}})-\left( 2(1-\varepsilon x\right) - \sqrt{S} \right) \cdot(-3x)\\
        &= \frac{3}{\sqrt{S}} \left( (1-x-\varepsilon x) ( -2x\sqrt{S} - (\varepsilon x^2 - x) ) - \sqrt{S} ( 2(1-\varepsilon x) - \sqrt{S} ) (-x) \right)\\
        &= \frac{3\cdot x^2\cdot(2 \sqrt{S}- ((3 - \varepsilon) x + 1))}{\sqrt{S}}.
    \end{align*}
    Our final step is to show $2 \sqrt{S}- ((3 - \varepsilon) x + 1)\geq0$. Notice that both $\sqrt{S}$ and $(3 - \varepsilon) x + 1$ are non-negative, so it is sufficient to prove that 
    \begin{align}
    \label{inq:last-inq-partial-x}
        (2 \sqrt{S})^2-((3 - \varepsilon) x + 1)^2\geq 0
    \end{align}
    And Inequality \eqref{inq:last-inq-partial-x} holds since
    \begin{align*}
        (2 \sqrt{S})^2-((3 - \varepsilon) x + 1)^2&=(2 \sqrt{(\varepsilon^2+3)x^2-2\varepsilon x+1})^2-((3 - \varepsilon) x + 1)^2\\
        &=4((\varepsilon^2+3)x^2-2\varepsilon x+1)-((\varepsilon^2-6\varepsilon+9)x^2+2(3-\varepsilon)x+1)\\
        &=(3\varepsilon^2+6\varepsilon+3)x^2-(6\varepsilon+6)x+3\\
        &= 3((\varepsilon+1)x-1)^2\geq 0.
    \end{align*}
\qed

\paragraph{Step 2: the PoA in terms of $k^\star$.}

We represent the Price of Anarchy by $x$, $\varepsilon$, and $k^\star$:
\begin{align*}
    PoA
    &=\frac{\varepsilon x - 2x - 1}{k^\star \left((\varepsilon x + x - 1)k^{\star2} - 3(\varepsilon x - 1)k^\star + 3(\varepsilon x - x - 1) \right)}.
\end{align*}

By Lemma \ref{lem:opt-2-mono}, we have
\begin{align}
\label{eq:opt-2-mono}
    x(3k^{\star2} - 1) = (1 - \varepsilon x)(3k^{\star2} - 4k^\star + 1).
\end{align}

Let $S=3k^{\star2}-4k^\star+1$. Equality \eqref{eq:opt-2-mono} implies
$(1 - \varepsilon x) = x\cdot\frac{3k^{\star2} - 1}{S}$.
Substituting into the expression of $PoA$, we have
    \begin{align*}
        PoA &= \frac{x\cdot(\frac{1-3k^{\star2}}{S}-2)}{x\cdot k^\star\cdot\left((\frac{1-3k^{\star2}}{S}+1)k^{\star2}-3(\frac{1-3k^{\star2}}{S})k^\star+3(\frac{1-3k^{\star2}}{S}-1)\right)}\\
        &=\frac{1-3k^{\star2}-2S}{k^\star((1-3k^{\star2}+S)k^{\star2}-3(1-3k^{\star2})k^\star+3(1-3k^{\star2}-S))}=-\frac{9k^{\star2}-8k^\star+1}{k^{\star2}(5k^{\star2}-16k^\star+9)}.
    \end{align*}
On $k^\star\in(1/3,\sqrt{3}/3)$ we have $9k^{\star2}-8k^\star+1<0$, so the expression above is positive, as a Price of Anarchy must be.

\paragraph{Step 3: monotonicity in $k^\star$.}

\begin{lemma}
    \label{lem:monotone-k-2}
    $-\dfrac{9k^{\star2}-8k^\star+1}{k^{\star2}(5k^{\star2}-16k^\star+9)}$ is monotonically decreasing with respect to $k^\star$ on $(1/3,\sqrt{3}/3)$.
\end{lemma}

\noindent\emph{Proof of Lemma \ref{lem:monotone-k-2}.}
 Define 
 \begin{align*}
     Q&=9k^{\star2}-8k^\star+1,\\
     R&=k^{\star2}(5k^{\star2}-16k^\star+9).
 \end{align*}
 The expression equals $-Q/R$, whose derivative is $\frac{1}{R^2}\big(Q\frac{\di R}{\di k^\star}-R\frac{\di Q}{\di k^\star}\big)$. Hence it is monotonically decreasing if, on $k^\star\in(1/3,\sqrt{3}/3)$,
 \begin{align*}
     Q\frac{\di R}{\di k^\star}-R\frac{\di Q}{\di k^\star}\leq 0.
 \end{align*}
 This inequality holds due to following derivation.
    \begin{align*}
        Q\frac{\di R}{\di k^\star}-R\frac{\di Q}{\di k^\star}&=(9k^{\star2} - 8k^\star + 1)(20k^{\star3} - 48k^{\star2} + 18k^\star) - (5k^{\star4} - 16k^{\star3} + 9k^{\star2})(18k^\star - 8)\\
        &= k^\star((180k^{\star4} - 592k^{\star3} + 566k^{\star2} - 192k^\star + 18)-(90k^{\star4} - 328k^{\star3} + 290k^{\star2} - 72k^\star))\\
        &= k^\star(90k^{\star4} - 264k^{\star3} + 276k^{\star2} - 120k^\star + 18)\\
        &=6k^\star(15k^{\star4}-44k^{\star3}+46k^{\star2}-20k^\star+3)\\
        &=6k^\star(k^\star-1)^2(3k^\star-1)(5k^\star-3)<0.
    \end{align*}
    The last inequality holds since $k^\star>0,(k^\star-1)^2>0,3k^\star-1>0,5k^\star-3<0$ on $k^\star\in(1/3,\sqrt{3}/3)$.
\qed

\paragraph{Conclusion.}

By Lemma \ref{lem:opt-2-mono}, $k^\star$ is monotonically increasing in both $x$ and $\varepsilon$, so $k^\star$ ranges over $(1/3,\sqrt{3}/3)$, attaining $1/3$ as $x\to 0,\varepsilon\to0$ and $\sqrt{3}/3$ as $x\to1,\varepsilon=1$. By Lemma \ref{lem:monotone-k-2}, the PoA is decreasing in $k^\star$, hence it is maximized at $k^\star=1/3$. Substituting $k^\star=1/3$ gives $PoA=27/19$. \qed

\section{Full Proof of Theorem \ref{thm:PoA-mono-lb}}
\label{appendix:proof-thm-mono-lb}

Consider an instance $I\in\mathcal{I}_{step}$ with demand
    \begin{align*}
        D(t) =
        \begin{cases}
        1, &t \in [0, \varepsilon] \\
        0, &t \in (\varepsilon, 1]
        \end{cases}
    \end{align*}
    and a monomial price function $P(z)=(d+1)\cdot z^d$ where integer $d\geq 1$ and  $G(z)=z^{d+1}$.

    The improvement on social cost is maximized when the net demand $N(t)=\varepsilon$ for all $t\in[0,1]$ due to Claim \ref{clm:cb-average}. We have
    \begin{align*}
        \mathrm{WEL}(B_{CB})=\varepsilon-\varepsilon^{d+1}.
    \end{align*}
    The decentralized battery deploys
    \begin{align*}
        B_{DCB}(t)=\begin{cases}
            x^\star, &t\in[0,\varepsilon]\\
            -\frac{\varepsilon x^\star}{1-\varepsilon},&t\in(\varepsilon,1]
        \end{cases},
    \end{align*}
    where
    $x^\star=\argmax_{x\in[0,1]}x\cdot((1-x)^{d}-(\frac{\varepsilon x}{1-\varepsilon})^d)$.
    Let $f(x)=x\cdot((1-x)^{d}-(\frac{\varepsilon x}{1-\varepsilon})^d)$. Then
    \begin{align*}
        f'(x)= (1-x)^{d-1} \left( 1 - (d+1)x \right) - (d+1) \left( \frac{\varepsilon x}{1-\varepsilon} \right)^d\leq(1-x)^{d-1} \left( 1 - (d+1)x \right).
    \end{align*}

    Since $f'(0)>0$ and $f'(x)\leq(1-x)^{d-1}(1-(d+1)x)<0$ for $x>1/(d+1)$, the maximizer satisfies $x^\star\leq\widebar{x}=1/(d+1)$. The price of anarchy satisfies
    \begin{align*}
        PoA&=\frac{\varepsilon-\varepsilon^{d+1}}{\varepsilon-\varepsilon\cdot G(1-x^\star)-(1-\varepsilon)G(\frac{\varepsilon x^\star}{1-\varepsilon})}\\
        &\geq \frac{\varepsilon-\varepsilon^{d+1}}{\varepsilon-\varepsilon\cdot G(1-x^\star)}\geq\frac{\varepsilon-\varepsilon^{d+1}}{\varepsilon-\varepsilon\cdot G(1-\widebar{x})}=\frac{1-\varepsilon^d}{1-(\frac{d}{d+1})^{d+1}}.
    \end{align*}
    When $\varepsilon\to0$, $PoA\geq \frac{1}{1-(\frac{d}{d+1})^{d+1}}$. As $d\to\infty$, this approaches $e/(e-1)$.
\qed

\section{Full Proof of Theorem \ref{thm:PoA-linear-n}}
\label{appendix:proof-thm-linear-n}

We prove that for $n$ identical batteries under linear pricing, the game admits a unique pure Nash equilibrium, which is symmetric, and $PoA=(n+1)^2/(n(n+2))$.

Consider $n$ batteries with operations $B_1,\ldots,B_n\in\mathcal{C}\overset{\text{def}}{=}\mathcal{B}\cap\Omega$. Under $P(z)=az+b$, the welfare improvement depends on the aggregate $S=\sum_i B_i$ via $\mathrm{WEL}=\tfrac{1}{2}a(2\langle D,S\rangle - \|S\|^2)$ (extending \eqref{eq:W-linear}), and each battery's revenue is $\mathrm{REV}_i=a(\langle D,B_i\rangle - \langle S,B_i\rangle)$ (extending \eqref{eq:REV-linear}). The revenue of battery $i$ can be rewritten as $\mathrm{REV}_i = a(\langle D - S_{-i}, B_i\rangle - \|B_i\|^2)$ where $S_{-i}=\sum_{j\neq i}B_j$.

Define the potential function $\Phi(B_1,\ldots,B_n)=\langle D,S\rangle - \tfrac{1}{2}\|S\|^2 - \tfrac{1}{2}\sum_{i=1}^n\|B_i\|^2$. One verifies that $\nabla_{B_i}\Phi = D - S - B_i = \frac{1}{a}\nabla_{B_i}\mathrm{REV}_i$, so $\Phi$ is an exact potential for the game. Writing $\Phi$ in vector form as $\Phi(\mathbf{B})=\langle D,\mathbf{1}^\top\mathbf{B}\rangle - \frac{1}{2}\mathbf{B}^\top(J+I)\mathbf{B}$ where $\mathbf{B}=(B_1,\ldots,B_n)^\top$, $J$ is the $n\times n$ all-ones matrix, and $I$ the identity, the matrix $J+I$ has eigenvalues $n+1$ (multiplicity $1$) and $1$ (multiplicity $n-1$), so it is positive definite. Therefore $\Phi$ is strictly concave on $\mathcal{C}^n$, which guarantees that the Nash equilibrium exists and is unique. By the permutation symmetry of the game, the unique equilibrium is symmetric: $B_i^*=B^*$ for all $i$, with $S^*=nB^*$. The centralized optimum is also symmetric: $B_i^{CB}=B_{CB}$ for all $i$, with $S_{CB}=nB_{CB}$.

At the symmetric equilibrium, each battery maximizes $\langle D-(n-1)B^*,B_i\rangle - \|B_i\|^2$ over $\mathcal{C}$. By Lemma \ref{lemma:variational_inequalities}, the first-order condition gives:
\begin{align}
\label{inq:opt-n-full}
    \langle D-(n+1)B^*,\;B - B^*\rangle \leq 0,\quad \forall B\in\mathcal{C}.
\end{align}

Substituting $B = B_{CB}$ into \eqref{inq:opt-n-full}:
\begin{align}
\label{inq:n-sub-cb}
    \langle D, B_{CB} - B^*\rangle \leq (n+1)\bigl(\langle B^*, B_{CB}\rangle - \|B^*\|^2\bigr).
\end{align}

Substituting $B = 0$ into \eqref{inq:opt-n-full} (recall $0\in\mathcal{C}$):
\begin{align}
\label{inq:n-sub-0}
    \langle D, B^*\rangle \geq (n+1)\|B^*\|^2.
\end{align}

Let $\alpha = (n+1)^2/(n(n+2))$. We show $\alpha g \geq f$ where $f=2\langle D,B_{CB}\rangle - n\|B_{CB}\|^2$ and $g=2\langle D,B^*\rangle - n\|B^*\|^2$. Compute:
\begin{align*}
    \alpha g - f &= 2(\alpha-1)\langle D,B^*\rangle - 2\bigl(\langle D,B_{CB}\rangle - \langle D,B^*\rangle\bigr) - \alpha n\|B^*\|^2 + n\|B_{CB}\|^2.
\end{align*}
Applying \eqref{inq:n-sub-cb} to bound the second term and \eqref{inq:n-sub-0} to bound the first:
\begin{align*}
    \alpha g - f &\geq 2(\alpha-1)(n+1)\|B^*\|^2 - 2(n+1)\langle B^*,B_{CB}\rangle + 2(n+1)\|B^*\|^2 - \alpha n\|B^*\|^2 + n\|B_{CB}\|^2\\
    &= \bigl[2\alpha(n+1) - \alpha n\bigr]\|B^*\|^2 - 2(n+1)\langle B^*,B_{CB}\rangle + n\|B_{CB}\|^2\\
    &= \alpha(n+2)\|B^*\|^2 - 2(n+1)\langle B^*,B_{CB}\rangle + n\|B_{CB}\|^2.
\end{align*}
With $\alpha = (n+1)^2/(n(n+2))$, we have $\alpha(n+2)=(n+1)^2/n$, so
\begin{align*}
    \alpha g - f &\geq \frac{(n+1)^2}{n}\|B^*\|^2 - 2(n+1)\langle B^*,B_{CB}\rangle + n\|B_{CB}\|^2 = n\left\|\frac{n+1}{n}B^* - B_{CB}\right\|^2 \geq 0.
\end{align*}
This establishes $PoA \leq (n+1)^2/(n(n+2))$.

To show tightness, consider the $\mathcal{I}_{\text{step}}$ instance with $D(t)=1$ for $t\in[0,1/2]$, $D(t)=0$ for $t\in(1/2,1]$, $P(z)=az$, and $\Omega$ non-binding. Write $\widetilde{D}=D-\widebar{D}$ for the mean-zero fluctuation, where $\widebar{D}=1/2$, so $\widetilde{D}(t)=1/2$ for $t\in[0,1/2]$ and $\widetilde{D}(t)=-1/2$ for $t\in(1/2,1]$. The centralized optimum assigns $B_{CB}=\widetilde{D}/n$ to each battery, yielding aggregate $S_{CB}=\widetilde{D}$ and constant net demand $N(t)\equiv 1/2$. We compute $\langle D,B_{CB}\rangle = \|\widetilde{D}\|^2/n = 1/(4n)$ and $\|B_{CB}\|^2=\|\widetilde{D}\|^2/n^2=1/(4n^2)$, giving $\mathrm{WEL}_{CB}=\frac{1}{2}a(2n\cdot\frac{1}{4n}-n^2\cdot\frac{1}{4n^2})=a/8$.

At the symmetric Nash equilibrium, the unconstrained first-order condition $D-(n+1)B^*=\lambda$ together with periodicity $\int_0^1 B^*\di t=0$ gives $\lambda=\widebar{D}=1/2$ and $B^*=\widetilde{D}/(n+1)$. We compute $\langle D,B^*\rangle = 1/(4(n+1))$ and $\|B^*\|^2=1/(4(n+1)^2)$, giving
\begin{align*}
    \mathrm{WEL}_{NE} = \frac{a}{2}\left(\frac{2n}{4(n+1)}-\frac{n^2}{4(n+1)^2}\right)=\frac{an(n+2)}{8(n+1)^2}.
\end{align*}
Therefore $PoA = (a/8) \big/ \bigl(an(n+2)/(8(n+1)^2)\bigr) = (n+1)^2/(n(n+2))$. \qed

\end{document}